\documentclass[a4paper, conference]{IEEEtran}
\usepackage[utf8]{inputenc}
\usepackage{amsmath,amssymb,amsthm}

\usepackage{bm}

\newtheorem{theorem}{Theorem}
\newtheorem*{lem}{Lemma}

\newtheorem*{cor}{Corollary}
\theoremstyle{remark}
\newtheorem*{rem}{Remark}

\newcommand{\mc}{\mathcal}
\newcommand{\mbb}{\mathbb}
\newcommand{\mbf}{\mathbf}

\title{Frequency hopping does not increase anti-jamming resilience of wireless channels}
\author{\IEEEauthorblockN{Moritz Wiese and Panos Papadimitratos}\IEEEauthorblockA{\\Networked Systems Security Group\\KTH Royal Institute of Technology, Stockholm, Sweden\\\{moritzw, papadim\}@kth.se}}

\begin{document}

\maketitle

\begin{abstract}
	The effectiveness of frequency hopping for anti-jamming protection of wireless channels is analyzed from an information-theoretic perspective. The sender can input its symbols into one of several frequency subbands at a time. Each subband channel is modeled as an additive noise channel. No common randomness between sender and receiver is assumed. It is shown that capacity is positive, and then equals the common randomness assisted (CR) capacity, if and only if the sender power strictly exceeds the jammer power. Thus compared to transmission over any fixed frequency subband, frequency hopping is not more resilient towards jamming, but it does increase the capacity. Upper and lower bounds on the CR capacity are provided.
\end{abstract}

\section{Introduction}

A wireless channel is open to inputs from anybody operating on the same frequency. Therefore communication has to be protected against deliberate jamming. This means that communication protocols have to be devised whose application enables reliable data transmission even if attacked by a jammer.

If a sufficiently broad frequency band is available, and if the jammer does not have simultaneous access to the complete band, a method which suggests itself is frequency hopping (FH). The frequency spectrum is divided into subbands. In each time slot, the sender chooses a subband in a random way and uses only that frequency to transmit data in that time slot. In some models \cite{EWFH,SPCCFH}, the receiver hops over frequencies, too, and only listens to one subband at a time. The idea is that in this way, the channel will not be jammed all the time with positive probability, and some information will go through.

To succeed, the basic FH idea requires common randomness known to sender and receiver, but unknown to the jammer. A careful analysis of that situation has been performed in \cite{EWFH}. It is clearly necessary that the common randomness realization be known before transmission starts. As the channel cannot be used to distribute this knowledge, this leads to a circle called anti-jamming/key-establishment dependency in \cite{SPCCFH}. 

In \cite{SPCCFH} it has been investigated for the first time whether FH can be used for data transmission without the availability of common randomness. Moreover, the jammer is allowed to distribute its power arbitrarily over all frequency subbands and use these simultaneously. It is assumed that whether the jammer inserts, modifies or jams messages only depends on the relation of its own and the sender's power. A protocol is found which achieves a positive throughput whose value depends on the jammer's strategies, e.g. whether or not it can listen to the sender's signals.

We take a different perspective in this work. The central figure of merit for our communication system is the message transmission error incurred under a jamming attack. A good FH protocol should make this error small. We assume that the jammer cannot listen to symbols sent through the channel (this in particular differs from \cite{SPCCFH}), that it knows the channel and the code, but not the specific message sent, and that it knows when the transmission of a new codeword begins. It can input symbols into any frequency subset of a given size. We also assume that the receiver listens to all frequencies simultaneously.

Within these boundaries, any jammer strategy is allowed. The jammer is successful if no coding strategy can be found making the transmission error vanish with increasing coding blocklength for any jamming strategy. This is an operational approach to measure the success of jamming, in contrast to the approach of \cite{SPCCFH} described above.

Using the information-theoretic model of an additive Arbitrarily Varying Channel (AVC) and the analysis in \cite{CsAVCgenalph}, we find that the success of a jammer indeed depends on the relation between its own and the sender's power. In fact, if the sender power is strictly larger than the jammer power, the same, positive capacity is achieved as in the case where sender and receiver have access to common randomness which is unknown to the jammer. If the converse relation between sender and jammer power holds, then no data transmission at all is possible. This is independent of the number $J$ of subchannels the jammer can influence at the same time.

On the other hand, it is known that for each frequency subband the same holds: If the jammer has more power than the sender, no communication is possible over this band, whereas the common randomness assisted capacity is achieved in case the sender power exceeds the jammer power. Thus in the case that no single frequency subband has a positive capacity without common randomness, then no FH scheme achieves a positive capacity either. Seen from this perspective, FH does not provide any additional protection against jamming compared to schemes which stick to one single frequency. However, FH does in general increase the common randomness assisted capacity compared to the use of one single subchannel, and hence also the capacity without common randomness if positive -- the FH sequence may depend on the message and thus reveal additional information. (In \cite{ZWLi,ZLii} this is called message-driven frequency hopping.) 

The common randomness assisted capacity will in general depend on the number $J$ of subchannels the jammer can simultaneously influence. Thus the capacity achievable without common randomness, if positive, also depends on $J$. We give a lower bound for the common randomness assisted capacity. If the noise is Gaussian and $J$ is sufficiently large, we also provide an upper bound which differs from the lower bound by the logarithm of the number of frequency bands. The bounds involve a waterfilling strategy for the distribution of the jammer's power over the frequencies.

\textit{Notation:} For any random variable $\xi$, we denote its distribution by $P_\xi$. The conditional distribution of a random variable $\xi$ given another random variable $\nu$ is denoted by $P_{\xi\vert\nu}$. 

\textit{Organization of the paper:} Section II presents the channel model and the main results. Sections III-VI contain the proofs of these results. A discussion concludes the paper in Section VII.

\section{System model and main results}

The total frequency band available for communication is divided into $K$ frequency subbands. These are modeled as parallel channels with additive noise. The receiver listens to all frequencies simultaneously. Frequency hopping (FH) means that the sender at each time instant chooses one of the $K$ subchannels into which it inputs a signal. For a fixed number $J$ with $1\leq J\leq K$, the jammer can at each time instant choose a subset $\mc I$ of the $K$ subchannels with $\lvert\mc I\rvert=J$ and input its own signals in subchannels belonging to this subset. 

The overall channel, called FH channel in the following, can be described as an additive Arbitrarily Varying Channel (AVC) with additive noise. For any $k\in\mc K=\{1,\ldots,K\}$, we set $(e_{k1},\ldots,e_{kK})^\top=\mbf e_k$ to be the vector with $e_{kk}=1$ and $e_{kl}=0$ for $l\neq k$. Further for any $\mc I$ with $\lvert\mc I\rvert=J$, we set $(e_{\mc I,1},\ldots,e_{\mc I,K})^\top=\mbf e_{\mc I}$ to be the vector satisfying $e_{\mc I,l}=1$ if $l\in\mc I$ and $e_{\mc I,l}=0$ else.

If the sender chooses symbol $x\in\mbb R$ to transmit over subchannel $k$, it inputs $x\mbf e_k$ into the channel. We denote the set $\mbb R\times\mc K$ by $\mc X$. The jammer choooses a subset $\mc I\subset\mc K$ of subchannels for possible jamming ($\lvert\mc I \rvert=J$) and a vector $(s_1,\ldots,s_K)^\top=\mbf s\in\mbb R^K$ of real numbers satisfying $s_l=0$ if $l\notin\mc I$. Then it inputs $\mbf s\circ\mbf e_{\mc I}$ into the channel, where the symbol $\circ$ denotes component-wise multiplication. We denote the set of possible jammer choices by
\[
	\mc S:=\bigcup_{\substack{\mc I\subset\mc K:\lvert\mc I\rvert=J}}\{\mc I\}\times\{\mbf s\in\mbb R^K:l\in\mc I\Rightarrow s_l=0\}
\]

The noise on different frequencies is assumed to be independent. Thus the noise probability distribution of the overall channel is determined by the noise distributions on the subchannels. For subchannel $k$, let $N_k$ be the noise random variable. Its mean is assumed to be zero and its variance is denoted by $\sigma_k^2$. The random vector $(N_1,\ldots,N_K)^\top$ is denoted by $\mbf N$. 

Given sender input $x\mbf e_k$ and jammer input $\mbf s\circ\mbf e_{\mc I}$, the receiver obtains a real $K$-dimensional output vector $(y_1,\ldots,y_K)^\top=\mbf y$ through the FH channel which satisfies
\[
	\mbf y=x\mbf e_k+\mbf s\circ\mbf e_{\mc I}+\mbf N.
\]
In particular, on frequencies without sender or jammer inputs, the output is pure noise. The channel is memoryless over time, i.e.\ outputs at different time instants are independent conditional on the sender and jammer inputs. Note that this is an additive AVC, but as its input alphabet is a strict subset of $\mbb R^K$, the special results of \cite{CsAVCgenalph} on additive-noise AVCs do not apply here. The general theory developed in \cite{CsAVCgenalph} is applicable, though: All alphabets involved are complete, separable metric spaces\footnote{Giving a discrete set $\mc K$ the metric $\rho(k,l)=1$ if $k\neq l$ and $\rho(k,k)=0$ for all $k,l\in\mc K$ makes $\mc K$ a complete metric space whose Borel algebra is its complete power set.}, the channel output distribution continuously depends on the sender and jammer inputs, and the constraints on sender and jammer inputs to be defined below are continuous. Hence the central hypotheses (H.1)-(H.4) of \cite{CsAVCgenalph} are satisfied.

The protocols used for data transmission are block codes. A blocklength-$n$ code is defined as follows. We assume without loss of generality that the set of messages $\mc M_n$ is the set $\{1,\ldots,\lvert\mc M_n\rvert\}$. An encoder is a mapping $f_n$ from $\mc M_n$ into the set of sequences of sender channel inputs of length $n$,
\[
	\{(x_1\mbf e_{k_1},\ldots,x_n\mbf e_{k_n}):(x_i,k_i)\in\mc X\;(1\leq i\leq n)\}.
\]
Note that this means that the sequence of frequency bands used by the sender may depend on the message to be sent. Every codeword can be considered as a $K\times n$-matrix whose $i$-th column is the $i$-th channel input vector. The decoder at blocklength $n$ is a mapping $\varphi_n:\mbb R^{K\times n}\longrightarrow\mc M_n$. 

Additionally, for some $\Gamma>0$, the sender has the power constraint $\sum_{i=1}^n\lVert f_n(m)_i\rVert^2\leq n\Gamma$ for all $m\in\mc M_n$, where $f_n(m)_i$ denotes the $i$-th column of the $K\times n$-matrix $f_n(m)$ and $\lVert\cdot\rVert$ denotes the Euclidean norm on $\mbb R^K$.  A code $(f_n,\varphi_n)$ with blocklength $n$ which satisfies the power constraint for $\Gamma$ is called an $(n,\Gamma)$-code.

We are interested in the transmission error incurred by a code $(f_n,\varphi_n)$. This error should be small for all possible jammer input sequences. Thus we first define the transmission error for a given length-$n$ jamming sequence $((\mc I_1,\mbf s_1),\ldots,(\mc I_n,\mbf s_n))$. This sequence can be given matrix form as well. We denote by $\tilde S$ the $K\times n$-matrix whose $i$-th column equals $\mbf s_i$. By $\tilde E\in\mbb R^{K\times n}$, we denote the matrix with columns $\mbf e_{\mc I_1},\ldots,\mbf e_{\mc I_n}$. Of course, $\tilde S\circ \tilde E=\tilde S$. We keep $\tilde E$ explicit because $\tilde S$ itself does not in general uniquely determine the sequence $(\mc I_1,\ldots,\mc I_n)$, as some components of $\mbf s_i$ could be zero $(1\leq i\leq n)$.

Just like the sender, the jammer has a power constraint. We require that $\sum_{i=1}^n\lVert\mbf s_i\rVert^2\leq n\Lambda$ for some $\Lambda>0$ and denote the set of $\tilde S\circ\tilde E$ satisfying this power constraint by $\mc J_\Lambda$. It is clear that a realistic jammer cannot transmit at arbitrarily large powers, so this is a reasonable assumption. Note that the jammer is free to distribute its power over the subchannel subset it has chosen for jamming. In particular, the power can be concentrated on one single frequency no matter what $J$ is. 

Now let $(f_n,\varphi_n)$ be a blocklength-$n$ code and $\tilde S\circ\tilde  E\in\mbb R^{K\times n}$ a jammer input. Then the average error incurred by $(f_n,\varphi_n)$ under this jamming sequence is defined to equal
\begin{multline*}
	\bar e(f_n,\varphi_n,\tilde S\circ\tilde E)\\=\frac{1}{\lvert\mc M_n\rvert}\sum_{m\in\mc M_n}\mbb P[\varphi_n(f_n(m)+\tilde S\circ\tilde  E+\tilde N)\neq m],
\end{multline*}
where $\tilde N$ is a matrix whose columns are $n$ independent copies of the noise random vector $\mbf N$. The overall transmission error for $(f_n,\varphi_n)$ under jammer power constraint $\Lambda$ is given by
\[
	\bar e(f_n,\varphi_n,\Lambda)=\sup_{\tilde S\circ\tilde E\in\mc J_\Lambda}\bar e(f_n,\varphi_n,\tilde S\circ\tilde E).
\]
This error criterion makes the FH channel an AVC.

A nonnegative real number is said to be an \textit{achievable rate} under sender power constraint $\Gamma$ and jammer power constraint $\Lambda$ if there exists a sequence of codes $((f_n,\varphi_n))_{n=1}^\infty$, where $(f_n,\varphi_n)$ is an $(n,\Gamma)$-code, satisfying
\begin{align*}
	\liminf_{n\rightarrow\infty}\frac{1}{n}\log\lvert\mc M_n\rvert&\geq R,\\
	\lim_{n\rightarrow\infty}\bar e(f_n,\varphi_n,\Lambda)&=0.
\end{align*}
The supremum $C(\Gamma,\Lambda)$ of the set of achievable rates under power constraints $\Gamma$ and $\Lambda$ is called the $(\Gamma,\Lambda)$-\textit{capacity} of the channel.

Now we ask under which conditions the $(\Gamma,\Lambda)$-capacity of the FH channel is positive, and in case it is positive, how large it is. A precise statement can be made upon introduction of the common randomness assisted capacity $C_r(\Gamma,\Lambda)$. This is the maximal rate achievable if sender and receiver have a common secret key unknown to the jammer. The key size is not restricted. As noted in the introduction, the presence of a certain amount of common randomness is a frequent assumption in the literature on frequency hopping.

For given power constraint $\Gamma>0$, we describe a common randomness assisted $(n,\Gamma)$-code as a random variable $(F_n,\Phi_n)$ on the set of $(n,\Gamma)$-codes with common message size and $(F_n,\Phi_n)$ independent of channel noise. The error it incurs under jamming sequence $\tilde S\circ\tilde E$ is defined to equal the mean $\mbb E[\bar e(F_n,\Phi_n,\tilde S\circ\tilde E)]$ over all possible realizations of $(F_n,\Phi_n)$, and the overall transmission error under jammer power constraint $\Lambda>0$ is set to equal
\[
	\sup_{\tilde S\circ\tilde E\in\mc J_\Lambda}\mbb E[\bar e(F_n,\Phi_n,\tilde S\circ\tilde E)].
\]
The definition of common randomness assisted achievable rate under power constraints $\Gamma$ and $\Lambda$ is now a straightforward extension of the corresponding notion for the deterministic case. The supremum of all common randomness assisted rates under power constraints $\Gamma$ and $\Lambda$ is called the common randomness assisted $(\Gamma,\Lambda)$-capacity and denoted by $C_r(\Gamma,\Lambda)$.

\begin{theorem}
	$C(\Gamma,\Lambda)$ is positive if and only if $\Gamma>\Lambda$. If it is positive, it equals $C_r(\Gamma,\Lambda)$.
\end{theorem}

\begin{cor}
	\begin{enumerate}
		\item If $C(\Gamma,\Lambda)>0$, then every fixed-frequency subchannel also has a positive capacity. In this sense FH is not necessary to achieve a positive rate. 
		\item If $C(\Gamma,\Lambda)>0$, then common randomness does not increase the maximal transmission rate.
	\end{enumerate}
\end{cor}

For $\Gamma>\Lambda$, it is thus desirable to have bounds on $C_r(\Gamma,\Lambda)$. These can be provided for all pairs $(\Gamma,\Lambda)$. Note that the choice of $\Lambda_1,\ldots,\Lambda_K$ below is a waterfilling strategy.

\begin{theorem}
\begin{enumerate}
	\item Let $\Lambda_1,\ldots,\Lambda_K$ be nonnegative numbers satisfying
	\[
		\begin{cases}
			\sigma_k^2+\Lambda_k=c	&\text{if }\sigma_k^2<c,\\
			\Lambda_k=0 &\text{if }\sigma_k^2\geq c
		\end{cases}
	\]
	with $c$ such that $\Lambda_1+\cdots+\Lambda_K=\Lambda$. Then
\begin{align}\label{eq:lower}
	C_r(\Gamma,\Lambda)\geq\frac{1}{2}\log\left(1+\frac{\Gamma}{c}\right).
\end{align}
	In particular, $C_r(\Gamma,\Lambda)>0$.
	\item If the noise is Gaussian and $J\geq\lvert\{k\in\mc K:\sigma_k^2<c\}\rvert$, then
	\begin{equation}\label{eq:upper}
		C_r(\Gamma,\Lambda)\leq\frac{1}{2}\log\left(1+\frac{\Gamma}{c}\right)+\log K.
	\end{equation}
\end{enumerate}
\end{theorem}

\begin{rem}
\begin{enumerate}
	\item Set $\mc K':=\{k\in\mc K:\sigma_k^2<c\}$. As comparison with \eqref{eq:upper} shows, \eqref{eq:lower} is a good bound if $J\geq\lvert\mc K'\rvert$ and the noise is Gaussian. The lack of a similar bound for the case $J<\lvert\mc K'\rvert$ can be explained by the fact that the jammer in this case has to leave some of the highest-throughput subchannels unjammed. $C_r(\Gamma,\Lambda)$ in general depends on $J$, and should increase for decreasing $J$. 
	\item The proof of Theorem 2 shows that the $\frac{1}{2}\log(1+\frac{\Gamma}{c})$  terms in \eqref{eq:lower}, \eqref{eq:upper} are achievable without frequency hopping, whereas frequency hopping contributes at most $\log K$ bits to capacity. According to the lower bound, the common randomness assisted capacity grows to infinity as $\Lambda$ is kept fixed and $\Gamma$ tends to infinity. Thus asymptotically for large $\Gamma$, the relative contribution to $C_r(\Gamma,\Lambda)$ of information transmitted through the FH sequence vanishes.
	\item Non-trivial frequency hopping will in general be necessary both to achieve $C_r(\Gamma,\Lambda)$ and $C(\Gamma,\Lambda)$. Although we will not prove this, this is implied by the mutual information characterization of $C_r(\Gamma,\Lambda)$ (see the proof of Theorem 2).
\end{enumerate}
\end{rem}

\section{Proof of Theorem 2}

Although Theorem 1 and its corollary are our main results, we first prove Theorem 2, which is needed for the proof of Theorem 1. From \cite[Theorem 4]{CsAVCgenalph} it follows that
\begin{align*}
	C_r(\Gamma,\Lambda)
	&=\sup_{\substack{(X,\kappa):\\\mbb E[X^2]\leq\Gamma}}\min_{\substack{(\iota,\mbf S):\\\mbb E[\lVert\mbf S\rVert^2]\leq\Lambda}}I(X\mbf e_\kappa;X\mbf e_\kappa+\mbf S\circ\mbf e_\iota+\mbf N)\\
	&=\min_{\substack{(\iota,\mbf S):\\\mbb E[\lVert\mbf S\rVert^2]\leq\Lambda}}\sup_{\substack{(X,\kappa):\\\mbb E[X^2]\leq\Gamma}}I(X\mbf e_\kappa;X\mbf e_\kappa+\mbf S\circ\mbf e_\iota+\mbf N).
\end{align*}
Here $X\mbf e\kappa$ is a random variable on the possible sender inputs determined by an $\mc X$-valued random pair $(X,\kappa)$. Similarly, $\mbf S\circ\mbf e_\iota$ is the jammer's random channel input determined by a random $\mc S$-valued pair $(\iota,\mbf S)$ independent of $(X,\kappa)$.

Define $\mbf Y=X\mbf e_\kappa+\mbf S\circ\mbf e_\iota+\mbf N$. The expression $I(X\mbf e_\kappa;\mbf Y)$ is concave in the distribution $P_\kappa$ of $\kappa$ and convex in the distribution $P_\iota$ of $\iota$. Therefore the sender will in general have to use frequency hopping to approach capacity and likewise, the jammer will not stick to one constant frequency subset $\mc I$ for jamming.

The mutual information term appearing in the above formula for $C_r(\Gamma,\Lambda)$ can be written as 
\begin{align}
	I(X\mbf e_\kappa;\mbf Y)
	&=I(X\mbf e_\kappa,\kappa;\mbf Y)-I(\kappa;\mbf Y\vert X\mbf e_\kappa)\notag\\
	&=I(X;\mbf Y\vert\kappa)+I(\kappa;\mbf Y),\label{eq:muti-expr}
\end{align}
upon application of the chain rule in each of the equalities and observing that the sequence $\kappa\leftrightarrow X\mbf e_\kappa\leftrightarrow\mbf Y$ is Markov. 

The second term in \eqref{eq:muti-expr} is between 0 and $\log K$. Thus to bound $C_r(\Gamma,\Lambda)$, it remains to bound 
\begin{align}\label{eq:Gleichheit}
	&\hphantom{\mathrel{=}}\min_{\substack{(\iota,\mbf S):\\\mbb E[\lVert\mbf S\rVert^2]\leq\Lambda}}\sup_{\substack{(X,\kappa):\\\mbb E[X^2]\leq\Gamma}}I(X;\mbf Y\vert\kappa)\\
	&=\min_{\substack{(\iota,\mbf S):\\\mbb E[\lVert\mbf S\rVert^2]\leq\Lambda}}\sup_{(\kappa,\bm\Gamma)}\sum_{k=1}^KP_\kappa(k)\sup_{X:\mbb E[X^2\vert\kappa=k]\leq\Gamma_k}I(X;\mbf Y\vert\kappa=k),\notag
\end{align}
where the supremum over $(\kappa,\bm\Gamma)$ is over $\kappa$ and nonnegative vectors $\bm\Gamma=(\Gamma_1,\ldots,\Gamma_K)$ satisfying $\sum P_\kappa(k)\Gamma_k\leq\Gamma$.

We continue with the proof of the lower bound. For any $k\in\mc K$,
\begin{equation}\label{eq:projection}
	I(X;\mbf Y\vert\kappa=k)\\\geq I(X;Y_k\vert\kappa=k).
\end{equation}
Fix any $(\iota,\mbf S)$ with $\mbb E[\lVert\mbf S\rVert^2]\leq\Lambda$. Let $S_{\mc I,k}$ be distributed according to the projection onto the $k$-th coordinate of $P_{\mbf S\vert\iota}[\cdot\vert\iota=\mc I]$ and denote the second moment of $S_{\mc I,k}$ by $\Lambda_{\mc I,k}$. Note that $\Lambda_{\mc I,k}=0$ if $k\notin\mc I$. The $k$-th coordinate output of the FH channel conditional on the event $\kappa=k$ has the form
\begin{equation}\label{eq:k-channel}
	y_k=x+Z_k,
\end{equation}
where $Z_k$ is a real-valued random variable whose distribution equals
\[
	P_{Z_k}=P_\iota(\{\mc I:k\notin\mc I\})P_{N_k}+\sum_{\mc I:k\in\mc I}P_\iota(\mc I)P_{N_k+S_{\mc I,k}}
\]
If we set $\Lambda_k:=\sum_{\mc I}P_\iota(\mc I)\Lambda_{\mc I,k}$, then $Z_k$ has the variance $\sigma_k^2+\Lambda_k$. Observe that $\Lambda_1+\cdots+\Lambda_K\leq\Lambda$. As \eqref{eq:k-channel} is an additive channel with the real numbers as input and output alphabet, it is a well-known fact \cite[Theorem 7.4.3]{Gallager} that
\[
	\sup_{\substack{X:\mbb E[X^2\vert\kappa=k]\leq\Gamma_k}}I(X;Y_k\vert\kappa=k)\\\geq\frac{1}{2}\log\left(1+\frac{\Gamma_k}{\sigma_k^2+\Lambda_{k}}\right).
\]
Hence the right-hand side of \eqref{eq:Gleichheit} can be lower-bounded by
\begin{equation}\label{eq:Gauss-absch}
	\min_{\bm\Lambda}\max_{(\kappa,\bm\Gamma)}\frac{1}{2}\sum_{k=1}^KP_\kappa(k)\log\left(1+\frac{\Gamma_k}{\sigma_k^2+\Lambda_{k}}\right),
\end{equation}
where the minimum is over vectors $\bm\Lambda=(\Lambda_1,\ldots,\Lambda_K)$ with nonnegative components satisfying $\Lambda_1+\cdots+\Lambda_K\leq\Lambda$. By choosing $\kappa$ to be constant and equal to the $k$ corresponding to the maximal $\log$-term in \eqref{eq:Gauss-absch} and by putting all power $\Gamma$ onto this $k$, \eqref{eq:Gauss-absch} is lower-bounded by
\begin{align}
	\min_{\Lambda_1+\cdots+\Lambda_K\leq\Lambda}\max_k\frac{1}{2}\log\left(1+\frac{\Gamma}{\sigma_k^2+\Lambda_{k}}\right).\label{eq:rhslower}
\end{align}
By this choice of $\kappa$, \eqref{eq:rhslower} is obtained without frequency hopping. It is now straightforward to show by comparison that waterfilling for the jammer is the optimal choice of $\Lambda_1,\ldots,\Lambda_K$ in \eqref{eq:rhslower}. This bound on \eqref{eq:Gleichheit} together with \eqref{eq:muti-expr} proves \eqref{eq:lower}.

Next we prove the upper bound \eqref{eq:upper}. Assume that all noise random variables are Gaussian. It is sufficient to upper-bound \eqref{eq:Gleichheit}. We are now free to choose any $(\iota,\mbf S)$ obeying the second moment condition. Thus let $\Lambda_1,\ldots,\Lambda_K$ satisfy the waterfilling scheme. Further, let $\mc I$ be a set containing $\mc K':=\{k:\sigma_k^2<c\}$ and choose $\iota$ to be constant and equal to this set (recall that $J\geq\lvert\mc K'\rvert$ by assumption). Define random variables $S_1,\ldots,S_K$, independent of each other and of the noise, by setting $S_k=0$ if $k\notin\mc K'$ and, for $k\in\mc K'$, by letting $S_k$ be Gaussian distributed with mean 0 and variance $\Lambda_k$.

The independence of $S_1,\ldots,S_K$ makes \eqref{eq:projection} an equality. Conditional on the event $\kappa=k$, the $k$-th coordinate output random variable is given by the formula
\[
	y_k=x+S_k+N_k,
\]
which is an additive Gaussian noise channel with noise variance $\sigma_k^2+\Lambda_k$. Applying \cite[Theorem 7.4.2]{Gallager}, we thus obtain
\[
	\sup_{X:\mbb E[X^2\vert\kappa=k]\leq\Gamma_k}I(X;Y_k\vert\kappa=k)=\frac{1}{2}\log\left(1+\frac{\Gamma_k}{\sigma_k^2+\Lambda_k}\right).
\]
So altogether, recalling the choice of $\Lambda_1,\ldots,\Lambda_K$, the right-hand side of \eqref{eq:Gleichheit} can be at most
\begin{align}
	&\sup_{(\kappa,\bm\Gamma)}\biggl\{\sum_{k\in\mc K'}P_\kappa(k)\frac{1}{2}\log\left(1+\frac{\Gamma_k}{c}\right)\notag\\
	&\qquad\qquad\qquad\qquad+\sum_{k\notin\mc K'}P_\kappa(k)\frac{1}{2}\log\left(1+\frac{\Gamma_k}{\sigma_k^2}\right)\biggr\}.\label{eq:endlich}
\end{align}
By replacing all $\sigma_k^2$ by $c$ and exploiting the concavity of the logarithm, one thus obtains that \eqref{eq:Gleichheit} is upper-bounded by
\begin{equation}\label{eq:nochendlicher}
	\frac{1}{2}\log\left(1+\frac{\Gamma}{c}\right),
\end{equation}
as claimed. Note that \eqref{eq:endlich} is equal to \eqref{eq:nochendlicher} if $\kappa$ is concentrated on one fixed $k\in\mc K'$ and the sender uses maximal power on this $k$, so \eqref{eq:nochendlicher} is also valid without frequency hopping. Note also that in the case of Gaussian noise and $J\geq\lvert\mc K'\rvert$, together with the lower bound proved before, we have thus obtained a closed-form characterization of \eqref{eq:Gleichheit}. This completes the proof of Theorem 2.

\section{Proof of direct part of Theorem 1}

The proof of Theorem 1 bases on the sufficient criterion for $C(\Gamma,\Lambda)=C_r(\Gamma,\Lambda)$ provided by the corollary to \cite[Theorem 4]{CsAVCgenalph}. To formulate this criterion, we first have to say what it means for the FH channel to be \textit{symmetrized} by a stochastic kernel.

A stochastic kernel $U$ with inputs from $\mc X$ and outputs in $\mc S$ gives, for every $(x,k)\in\mc X$, a probability measure $U(\cdot\vert x,k)$ on the Borel algebra of $\mc S$ such that for every Borel-measurable $\mc A\subset\mc S$, the mapping $(x,k)\mapsto U(\mc A\vert x,k)$ is measurable. $U(\cdot\vert x,k)$ is specified by its values on all pairs $(\mc I,\mc B)$, where $\lvert\mc I\rvert=J$ and $\mc B$ is a Borel set on $\mbb R^K$ such that for all $\mbf b\in\mc B$, it holds that $l\notin\mc I$ implies $b_l=0$. One can thus write
\[
	U(\mc I,\mc B\vert x,k)=U_1(\mc I\vert x,k)U_2(\mc B\vert x,k,\mc I).
\]
$U_1(\cdot\vert x,k)$ determines a random variable $\iota^U(x,k)$ on the set of subsets of $\mc K$ with cardinality $J$. $U(\cdot\vert x,k)$ then determines a random variable $\mbf S^U(x,k)$ which, conditional on the event $\iota^U(x,k)=\mc I$, has the distribution $U_2(\cdot\vert x,k,\mc I)$. These random variables give rise to a random jammer input, $Z^U_{x,k}:=\mbf S^U(x,k)\circ\mbf e_{\iota^U(x,k)}$. Thus any pair $(x',k')\in\mc X$ together with $U$ defines the following channel:
\[
	\mbf y=x\mbf e_k+\mbf Z_{x',k'}^U+\mbf N,
\] 
where $(x,k)\in\mc X$ is the sender input, the output set is $\mbb R^K$, and the noise is $\mbf Z_{x',k'}^U+\mbf N$.

By definition, the FH channel is \textit{symmetrized} by $U$ if all sender input pairs $(x,k)$ and $(x',k')$ satisfy
\[
	x\mbf e_k+\mbf Z_{x',k'}^U+\mbf N\stackrel{\mc D}{=}x'\mbf e_{k'}+\mbf Z_{x,k}^U+\mbf N,
\] 
where $\stackrel{\mc D}{=}$ means that the left-hand and the right-hand side have the same distribution. In particular, this implies
\[
	x\mbf e_k+\mbb E\bigl[\mbf Z_{x',k'}^U+\mbf N\bigr]=x'\mbf e_{k'}+\mbb E\bigl[\mbf Z_{x,k}^U+\mbf N\bigr]
\] 
or equivalently, as the noise is mean-zero,
\begin{equation}\label{eq:symmetr}
	x\mbf e_k+\mbb E\bigl[\mbf Z_{x',k'}^U\bigr]=x'\mbf e_{k'}+\mbb E\bigl[\mbf Z_{x,k}^U\bigr].
\end{equation}

To state the criterion for the equality of the $(\Gamma,\Lambda)$-capacities with and without common randomness, some more definitions are necessary. Let $\mc U_0$ be the class of stochastic kernels $U$ that symmetrize the FH channel and for which $\mbf Z_{x,k}^U$ has finite variance for all $(x,k)$. Let $\tilde{\mc X}\subset\mc X$ be finite and $(X,\kappa)$ be concentrated on $\tilde{\mc X}$. Assume that for every $(x,k)\in\mc X$, the conditional distribution of the random variable $Z_{X,\kappa}^U$ given $\{X=x,\kappa=k\}$ equals that of $Z_{x,k}^U$. Then define
\[
	\tau_{\tilde{\mc X}}(X,\kappa,\Lambda)=\frac{1}{\Lambda}\inf_{U\in\mc U_0}\mbb E\bigl[\lVert\mbf Z_{X,\kappa}^U\rVert^2\bigr].
\]
We also write $C_{r,\tilde{\mc X}}(\Gamma,\Lambda)$ for the common randomness assisted capacity of the FH channel with the same power constraints, but whose inputs are restricted to the finite subset $\tilde{\mc X}$ of $\mc X$.

By the corollary of \cite[Theorem 4]{CsAVCgenalph}, $C(\Gamma,\Lambda)=C_r(\Gamma,\Lambda)$ if there exists a family $\mc F$ of finite subsets of $\mc X$ satisfying that every finite subset of $\mc X$ is contained in some member of $\mc F$ and that for every $\tilde{\mc X}\in\mc F$, there is an $(X,\kappa)$ concentrated on $\tilde{\mc X}$ and satisfying $\mbb E[X^2]\leq\Gamma$ with $I(X\mbf e_\kappa;\mbf Y)=C_{r,\tilde{\mc X}}(\Gamma,\Lambda)$ and $\tau_{\tilde{\mc X}}(X,\kappa,\Lambda)>1$.

We will now closely follow the proof of \cite[Theorem 5]{CsAVCgenalph} to prove that the above criterion is satisfied for the FH channel if $\Gamma>\Lambda$. Fix $\Gamma,\Lambda>0$. Let $\tilde{\mc X}_0$ be a finite set satisfying $C_{r,\tilde{\mc X}_0}(\Gamma',\Lambda)>C_r(\Gamma,\Lambda)$ for some $\Gamma'>\Gamma$. Such a set exists by the fact (\cite[Theorem 4]{CsAVCgenalph}) that for all $\Gamma,\Lambda$,
\[
	C_r(\Gamma,\Lambda)=\sup_{\tilde{\mc X}\subset\mc X\text{ finite}}C_{r,\tilde{\mc X}}(\Gamma,\Lambda)
\]
and the lower bound on $C_r(\Gamma,\Lambda)$ of Theorem 2 showing that $C_r(\Gamma,\Lambda)$ tends to infinity as $\Lambda$ is fixed and $\Gamma$ tends to infinity. We choose $\mc F$ as the family of finite subsets $\tilde{\mc X}$ of $\mc X$ satisfying $\tilde{\mc X}_0\subset\tilde{\mc X}$ and 
\[
	\tilde{\mc X}=\bigcup_{k=1}^K\tilde{\mc X}_k\times\{k\},
\]
where $\tilde{\mc X}_k$ is symmetric about the origin. Obviously, every finite subset of $\mc X$ is contained in some $\tilde{\mc X}\in\mc F$. We first need to show that for every finite input set $\tilde{\mc X}\in\mc F$ there exist $C_{r,\tilde{\mc X}}(\Gamma,\Lambda)$-achieving channel input distributions which exhaust all the power and are symmetric on every frequency subband.

\begin{lem}\label{lem:finiteP}
	Let $\tilde{\mc X}\in\mc F$. Then there exists a pair $(X,\kappa)$ of random variables with values in $\tilde{\mc X}$ satisfying 
	\begin{equation}\label{eq:endlopt}
		\min_{\substack{(\iota,\mbf S):\\\mbb E[\lVert\mbf S\rVert^2]\leq\Lambda}}I(X\mbf e_\kappa;X\mbf e_\kappa+\mbf S\circ\mbf e_\iota+\mbf N)=C_{r,\tilde{\mc X}}(\Gamma,\Lambda)
	\end{equation}
	and
	\begin{align}
		\mbb E\bigl[X^2\bigr]&=\Gamma,\label{eq:fullpower}\\
		P_{X\vert\kappa}(\cdot\vert k)&=P_{-X\vert\kappa}(\cdot\vert k)\qquad(1\leq k\leq K)\label{eq:symmetry}
	\end{align}
	Here $P_{X\vert\kappa}$ denotes the conditional probability of $X$ given $\kappa$, and $P_{-X\vert\kappa}$ is defined analogously.
\end{lem}

\begin{proof}
	Fix $\tilde{\mc X}\in\mc F$. By definition of $\tilde{\mc X}_0$ and $\Gamma'$ we have 
	\[
		C_{r,\tilde{\mc X}}(\Gamma',\Lambda)\geq C_{r,\tilde{\mc X}_0}(\Gamma',\Lambda)>C_{r,\tilde{\mc X}}(\Gamma,\Lambda).
	\]
	Let $(X,\kappa)$ and $(X',\kappa')$ assume values in $\tilde{\mc X}$ such that $(X,\kappa)$ achieves $C_{r,\tilde{\mc X}}(\Gamma,\Lambda)$ and $(X',\kappa')$ achieves $C_{r,\tilde{\mc X}}(\Gamma',\Lambda)$. Then any $(\tilde X,\tilde\kappa)$ distributed according to a nontrivial convex combination of $P_{(X,\kappa)}$ and $P_{(X',\kappa')}$ achieves a rate larger than $C_{r,\tilde{\mc X}}(\Gamma,\Lambda)$ because the left-hand side of \eqref{eq:endlopt} is concave in $P_{(X,\kappa)}$ by \cite[Lemma 5]{CsAVCgenalph}. Moreover $(X',\kappa')$ uses strictly more power than $(X,\kappa)$, so the second moment of $X$ must equal $\Gamma$. This proves \eqref{eq:fullpower}.
	
	If we replace $(X,\kappa)$ by $(-X,\kappa)$, then the left-hand side of \eqref{eq:endlopt} remains unchanged. This is due to the symmetry of the jammer input constraints. Hence a random input $(\tilde X,\tilde\kappa)$ distributed according to $\frac{1}{2}(P_{(X,\kappa)}+P_{(-X,\kappa)})$ satisfies \eqref{eq:endlopt}-\eqref{eq:symmetry}.
\end{proof}

Let $\tilde{\mc X}\in\mc F$ and $(X,\kappa)$ as in the Lemma. We now show that $\tau_{\tilde{\mc X}}(X,\kappa,\Lambda)>1$ if $\Gamma>\Lambda$. To do so, choose any $U\in\mc U_0$. Then for any $(x',k')\in\mc X$, using Jensen's inequality,
\begin{align}
	&\mathrel{\hphantom{=}}\mbb E\bigl[\lVert\mbf  Z_{X,\kappa}^U\rVert^2\bigr]\notag\\
	%&=\mbb E\bigl[\mbb E\bigl[\lVert\mbf  Z_{X,\kappa}^U\rVert^2\vert X,\kappa\bigr]\bigr]\notag\\
	&=\sum_{(x,k)\in\tilde{\mc X}}P_{(X,\kappa)}(x,k)\mbb E\bigl[\lVert\mbf  Z_{x,k}^U\rVert^2\bigr]\notag\\
	&\geq\sum_{(x,k)\in\tilde{\mc X}}P_{(X,\kappa)}(x,k)\left\lVert\mbb E\bigl[\mbf Z_{x,k}^U\bigr]\right\rVert^2\label{eq:symmetr1}.
\end{align}
As $U$ symmetrizes the FH channel, we can apply \eqref{eq:symmetr} and lower-bound \eqref{eq:symmetr1} by
\begin{align}
	&\sum_{(x,k)\in\tilde{\mc X}}P_{(X,\kappa)}(x,k)\left\lVert x\mbf e_k-x'\mbf e_{k'}+\mbb E[\mbf Z_{x',k'}^U]\right\rVert^2\notag\\
	&\geq\sum_kP_\kappa(k)\sum_{x\in\tilde{\mc X}_k}P_{X\vert\kappa}(x\vert k)\lvert x-x'e_{k'k}+\mbb E[Z_{x',k'}^U(k)]\rvert^2,\label{eq:last}
\end{align}
where we denote by $Z_{x,k}^U(k)$ the $k$-th component of $\mbf Z_{x,k}^U$. By \eqref{eq:symmetry}, $P_{X\vert\kappa}(\cdot\vert k)$ is symmetric for every $k$, so its mean equals 0 and 
\[
	\min_a\sum_{x\in\tilde{\mc X}_k}P_{X\vert\kappa}(x\vert k)\lvert x-a\rvert^2=\sum_{x\in\tilde{\mc X}_k}P_{X\vert\kappa}(x\vert k)\lvert x\rvert^2.
\]
Using this in \eqref{eq:last} and applying \eqref{eq:fullpower} yields the lower bound
\[
	\sum_{(x,k)\in\tilde{\mc X}}P(x,k)\lvert x\rvert^2=\mbb E[X^2]=\Gamma
\]
for $\mbb E\bigl[\lVert\mbf  Z_{X,\kappa}^U\rVert^2\bigr]$. We conclude that $\tau_{\tilde{\mc X}}(X,\kappa,\Lambda)>1$ for all $\tilde{\mc X}\in\mc F$ and the corresponding $(X,\kappa)$ if $\Gamma>\Lambda$, implying that $C(\Gamma,\Lambda)=C_r(\Gamma,\Lambda)$. As the common randomness assisted $(\Gamma,\Lambda)$-capacity is positive for positive $\Gamma$, this further implies that $C(\Gamma,\Lambda)>0$ if $\Gamma>\Lambda$, and the proof of the direct part of Theorem 1 is complete.

\section{Proof of converse for Theorem 1}

The converse follows the lines of the proof of the converse of \cite[Theorem 1]{CNAVCGauss}. Let $\Gamma\leq\Lambda$.  Let $(f_n,\varphi_n)$ be any $(n,\Gamma)$-code with $\lvert\mc M_n\rvert\geq 2$ messages. We will prove that there exists a jammer input sequence $\tilde S\circ\tilde E$ such that $\bar e(f_n,\varphi_n,\tilde S\circ\tilde E)\geq 1/4$. This sequence will be found among the following inputs. Assume that $f_n(m)=(x_1\mbf e_{k_1},\ldots,x_n\mbf e_{k_n})$. Then let $\tilde E(m)$ be the matrix whose $i$-th column is $\mbf e_{k_i}$, and let the $i$-th column of the matrix $\tilde S(m)$ equal $x_i\mbf e_{k_i}$. This gives a set $\{\tilde S(m)\circ\tilde E(m):m\in\mc M_n\}$ of jammer input sequences. Note that the power of any of these is at most $\Lambda$.

Observe that for $m,m'\in\mc M_n$ with $m\neq m'$,
\begin{multline*}
	\mbb P[\varphi_n(f_n(m)+\tilde S(m')\circ\tilde E(m')+\tilde{N})\neq m]\\
	=\mbb P[\varphi_n(f_n(m')+\tilde S(m)\circ\tilde E(m)+\tilde{N})\neq m]\\
	\geq 1-P[\varphi_n(f_n(m')+\tilde S(m)\circ\tilde E(m)+\tilde{N})\neq m'].
\end{multline*}
Thus
\begin{align*}
	&\mathrel{\hphantom{\geq}}\frac{1}{\lvert\mc M_n\rvert}\sum_{m\in\mc M_n}\bar e(f_n,\varphi_n,\tilde S(m)\circ\tilde E(m))\\
	&\geq\!\frac{1}{\lvert\mc M_n\rvert^2}\!\sum_{m,m'\in\mc M_n}\!\!\!\mbb P[\varphi_n(f_n(m)+\tilde S(m')\!\circ\!\tilde E(m')+\tilde{N})\neq m]\\
	&\geq\frac{1}{\lvert\mc M_n\rvert^2}\cdot\frac{\lvert\mc M_n\rvert(\lvert\mc M_n\rvert-1)}{2}\geq\frac{1}{4}.
\end{align*}
Therefore one of the jammer inputs $\tilde S(m)\circ\tilde E(m)$ makes the average error incurred by the code $(f_n,\varphi_n)$ at least one quarter. This proves the converse of Theorem 1.

\section{Proof of the corollary to Theorem 1}

The second claim of the corollary is obvious from Theorem 1. The first statement follows from \cite[Theorem 5]{CsAVCgenalph}, which says that an additive-noise channel with $\mbb R$ as sender, jammer and output alphabet has positive capacity (then equal to the common randomness assisted capacity) if and only if the sender power exceeds the jammer power. So if both the sender and the jammer in the FH channel concentrate their power on any frequency band $k\in\mc K$ and $\Gamma>\Lambda$, already a positive capacity equal to 
\[
	\max_{X:\mbb E[X^2]\leq\Gamma}\min_{S:\mbb E[S^2]\leq\Lambda}I(X;X+S+N_k)
\]
and lower-bounded by 
\[
	\frac{1}{2}\log\left(1+\frac{\Gamma}{\sigma_k^2+\Lambda}\right)>0
\]
will be achievable. In particular, this rate can be obtained without frequency hopping. On the other hand, if no transmission is possible over the subchannels, then $\Gamma\leq\Lambda$, and the FH channel also has zero capacity.

\section{Discussion}\label{sect:disc}

For non-discrete AVCs, there is no general statement that capacity without common randomness always equals 0 or the common randomness assisted capacity like the Ahlswede dichotomy in \cite{A1} for discrete AVCs. Thus it is not possible to justify Theorem 1 just by observing that the capacity of every subchannel is positive if $\Gamma>\Lambda$.

Like \cite{ZWLi,ZLii} we assume here that the receiver simultaneously listens on all frequencies. A different approach is taken in \cite{SPCCFH,EWFH}, where the receiver listens randomly on only one frequency band at a time. The above analysis can be performed in a similar way for this situation and leads to analogous results: The capacity without common randomness shared between sender and receiver is positive if and only if the sender power exceeds the jammer power. Of course, the capacity will in general be smaller than if the receiver listens on all frequencies.

The converse shows that in order to find a good jamming sequence, the jammer needs knowledge of the channel and the transmission protocol. Further, it should know when the transmission of a codeword starts, so it has to be synchronized with the sender. If this is given, then the successful jamming strategy in the case $\Gamma\leq\Lambda$ is to confuse the receiver: There exists a legitimate codeword such that if the jammer inputs this into the FH channel, the receiver cannot distinguish the sender's messages. 

The case of a jammer listening to the sender's input into the channel like in \cite{SPCCFH,EWFH} was not treated here because there exist few results on AVCs in this direction.

\end{document}